\documentclass[journal]{IEEEtran}

\IEEEoverridecommandlockouts
\usepackage[dvips]{graphicx}
\usepackage{algorithm}
\usepackage{algorithmic}
\usepackage{xcolor}
\usepackage{float}
\usepackage[ancient]{flushend} 
\usepackage{psfrag}
\newcommand{\norm}[1]{\left\lVert#1\right\rVert}
\usepackage{bm}
\usepackage{amssymb,mathtools}
\usepackage{amsfonts}
\usepackage{cite}
\usepackage{amsmath,amssymb}
\usepackage{amsthm} 
\usepackage{subfigure}
\newtheorem{theo}{Theorem}
\newtheorem{lem}{Lemma}

\theoremstyle{definition}

\ifCLASSINFOpdf

\else

\fi

\begin{document}

\title{Decentralized Optimization with Distributed Features and Non-Smooth Objective Functions}

\author{Cristiano~Gratton, 
        Naveen~K.~D.~Venkategowda,~\IEEEmembership{Member,~IEEE,}
		Reza~Arablouei,
        and~Stefan~Werner,~\IEEEmembership{Senior~Member,~IEEE}
\thanks{This work was partly supported by the Research Council of Noway. A conference precursor of this work appears in the \emph{Proceedings of the European Signal Processing Conference}, Dublin, Ireland, August 2021~\cite{Gratton2021eusipco}.}	
\thanks{C. Gratton and S. Werner are with the Department of Electronic Systems, Norwegian University of Science and Technology, Trondheim, Norway (email:cristiano.gratton@ntnu.no; stefan.werner@ntnu.no).}
\thanks{N. K. D. Venkategowda is with the Department of Science and Technology, Linköping University, Norrköping, Sweden (email: naveen.venkategowda@liu.se)}
\thanks{R. Arablouei is with the Commonwealth Scientific and Industrial Research Organisation, Pullenvale QLD 4069, Australia (email:reza.arablouei@csiro.au).}}

\maketitle

\begin{abstract}
We develop a new consensus-based distributed algorithm for solving learning problems with feature partitioning and non-smooth convex objective functions. Such learning problems are not separable, i.e., the associated objective functions cannot be directly written as a summation of agent-specific objective functions. To overcome this challenge, we redefine the underlying optimization problem as a dual convex problem whose structure is suitable for distributed optimization using the alternating direction method of multipliers (ADMM). Next, we propose a new method to solve the minimization problem associated with the ADMM update step that does not rely on any conjugate function. Calculating the relevant conjugate functions may be hard or even unfeasible, especially when the objective function is non-smooth. To obviate computing any conjugate function, we solve the optimization problem associated with each ADMM iteration in the dual domain utilizing the block coordinate descent algorithm.
Unlike the existing related algorithms, the proposed algorithm is fully distributed and does away with the conjugate of the objective function. We prove theoretically that the proposed algorithm attains the optimal centralized solution. We also confirm its network-wide convergence via simulations.
\end{abstract}


\IEEEpeerreviewmaketitle

\section{Introduction}

\IEEEPARstart{P}{erforming} data analytic tasks at a central processing unit in a distributed network can be infeasible due to the associated computing/communication costs or privacy issues. In addition, collecting all the data in a central hub creates a single point of failure. Therefore, it is necessary to develop algorithms that facilitate in-network processing and model learning using data collected by nodes/agents that are dispersed over a network \cite{Mingyihongbook,Grattonasilomar2018,Gratton2019,Giannakis2016,Hajinezhad2019,Nedic2009}. Distributed optimization problems pertain to several applications in statistics \cite{Mingyihongbook,Grattonasilomar2018}, signal processing \cite{Gratton2019,Giannakis2016}, machine learning, and control~\cite{Hajinezhad2019,Nedic2009}. 

An essential aspect of distributed learning is how the data is distributed among the agents that determines what each agent intends to or is able to learn. Horizontal partitioning of data refers to the case when the data samples containing all features are distributed over the network. That is, all the agents estimate the same common model. Examples of learning with horizontal partitioning of data can be found in \cite{Mingyihongbook,Gratton2019,Gratton2020eusipco,Bertrand2011}. On the other hand, when subsets of the features of all data samples are distributed over the network agents, we have feature partitioning of the data and every agent estimates a local model that is a part of the network-wide model. In the machine learning terminology, features are the descriptors or measurable characteristics of the data samples. In regression analysis, they may be called predictors or independent explanatory variables. 

Several machine learning problems deal with heterogeneous distributed data that common features cannot describe. For example, in multi-agent systems, each agent may acquire data to learn a local model and refrains from sharing the data with other agents due to resource constraints or privacy concerns. However, the aggregate data can be exploited to enhance accuracy or augment inference due to the correlation of the data across agents. In the Internet of things, a device may only be interested in estimating its own local model parameters. However, multiple devices distributed over an ad hoc network may be able to collectively process the network-wide data and enhance the estimation/inference quality. Distributed learning problems with feature partitioning arise in several signal processing applications, e.g., bioinformatics, multi-view learning, and dictionary learning, as mentioned in \cite{Boyd2010,Ying2019}. Data with feature partitioning can also be referred to as attribute-distributed data \cite{Zheng2011}, vertically-partitioned data \cite{Mangasarian,Vaidya}, data with column-partitioning \cite{Mota2012}, or heterogeneous data~\cite{Zheng2011}.

\subsection{Related Works}


Learning problems with feature partitioning of data have been considered in \cite{Mota2012,Leus2018,Mota2013,Boyd2010,Kashyap2016,Heinze2015,dualloco,Heinze2017,Ying2019,Sayed_dictionary,Sayed_eusipco,Arablouei2015main,Virginiasmith,dc_admm,Grattonasilomar2018,Gratton2021eusipco,CISS2018,Diniu2019}. The algorithms proposed in \cite{Mota2012,Leus2018} solve the basis pursuit and lasso problems, respectively. The work of \cite{Mota2013} assumes an appropriate coloring scheme of the network and cannot be extended to a general graph labeling.

The algorithms proposed in \cite{Boyd2010,dualloco,Kashyap2016,Heinze2015,Heinze2017} are not fully distributed since their consensus constraints involve the entire network instead of each agent’s local neighborhood. Furthermore, the algorithms proposed in \cite{Kashyap2016,Heinze2015} only solve the ridge regression problem, while the works of \cite{dualloco,Heinze2017} assume the cost function to be convex and smooth with Lipschitz-continuous gradient. Both algorithms proposed in \cite{dualloco,Heinze2017} can only be used for minimization problems with $\ell_2$-norm regularization (ridge penalty) and rely on the computation of the conjugate of the cost function.

The algorithms in \cite{Ying2019,Sayed_dictionary,Sayed_eusipco,Arablouei2015main} are based on the diffusion strategy, which is suitable when stochastic gradients are available. Furthermore, the work in \cite{Ying2019} assumes that the cost function is convex and smooth with Lipschitz-continuous gradient. The algorithm developed in \cite{Sayed_dictionary} relies on the calculation of the relevant conjugate functions. In addition, it assumes that the cost function is convex and smooth, and the regularization functions are strongly convex. The work of \cite{Sayed_eusipco} also assumes that the regularizer functions are smooth and strongly convex. Moreover, it relies on the computation of conjugate functions similar to the algorithms proposed in \cite{Virginiasmith,dc_admm}.
The diffusion-based algorithm proposed in \cite{Arablouei2015main} only solves the ridge regression problem. 
The consensus-based algorithm of \cite{Grattonasilomar2018} is also designed for ridge regression with feature partitioning. It outperforms the algorithm proposed in \cite{Arablouei2015main} in terms of convergence speed.

The algorithm proposed in \cite{Gratton2021eusipco} is designed for an $\ell_2$-norm-square cost function and hence cannot be extended to general objective functions. The works of \cite{Szurley2017,Chen2014,Berberidis2014,Berberidis2015} consider distributed agent-specific parameter estimation problem. However, the objective functions considered in these works are smooth.
The authors of \cite{CISS2018} propose a distributed coordinate-descent algorithm to reduce the communication cost in distributed learning with feature partitioning. However, the cost function in \cite{CISS2018} is assumed to be strongly convex and smooth. The work of \cite{Diniu2019} considers an asynchronous stochastic gradient-descent algorithm for learning with distributed features. However, the objective function in \cite{Diniu2019} is assumed to be smooth.

None of the above-mentioned existing algorithms for distributed learning with feature partitioning is designed for optimizing generic non-smooth objective functions over arbitrary graphs without using or computing any conjugate function.

\begin{table}[t!]
	\caption{Comparative Summary}
	\centering
	\begin{tabular}{|c|c|c|c|c|}
		\hline
		& \vtop{\hbox{\strut fully}\hbox{\strut distributed}} & \vtop{\hbox{\strut non-smooth}\hbox{\strut cost function}}&\vtop{\hbox{\strut non-smooth}\hbox{\strut regularizer}} & \vtop{\hbox{\strut no conjugate}\hbox{\strut function}}   \\  
		\hline 
	  \cite{Gratton2021eusipco}      &    \checkmark   &     & \checkmark       &       \\ 
		\hline 
		 \cite{Grattonasilomar2018}  &  \checkmark     &      &       &       \\
		\hline 
		\cite{Boyd2010}              &       &    \checkmark   &  \checkmark     &     \\
		\hline 
		\cite{Ying2019}              &  \checkmark     &       &      & \checkmark      \\
		\hline 
		\cite{Mota2012}              &   \checkmark    &     \checkmark  &       &       \\
		\hline 
		\cite{Leus2018}              &      &      &     \checkmark  & \checkmark      \\
		\hline		
		\cite{Mota2013}              &  \checkmark     &     &     \checkmark  &      \\
		\hline 
		\cite{Kashyap2016}           &       &      &      & \checkmark      \\
		\hline	
		\cite{dualloco}              &       &       &       &     \\
		\hline
		\cite{Heinze2015}            &       &       &     & \checkmark \\
		\hline
		\cite{Heinze2017}            &       &     &     &      \\
		\hline
		\cite{Sayed_dictionary}      &       &     \checkmark  &     \checkmark  & \checkmark      \\
		\hline 
		\cite{Sayed_eusipco}         & \checkmark      &   \checkmark   &       &     \\
		\hline 		
		\cite{Arablouei2015main}     &  \checkmark     &      &       & \checkmark      \\
		\hline	
		\cite{Virginiasmith}         &       &       &     \checkmark  &       \\
		\hline
		\cite{dc_admm}               &   \checkmark    &    &     \checkmark  &      \\
		\hline		
		\cite{CISS2018}              &       &      &     \checkmark  & \checkmark      \\
		\hline
		\cite{Diniu2019}             &       &      &    & \checkmark      \\
		\hline
		\cite{Szurley2017}           & \checkmark  &      &    & \checkmark      \\
		\hline
		\cite{Chen2014}              & \checkmark  &      &    & \checkmark      \\
		\hline	
		\cite{Berberidis2014}        & \checkmark  &      &    & \checkmark      \\
		\hline
		\cite{Berberidis2015}        & \checkmark  &      &    & \checkmark      \\
		\hline		
		proposed                     &    \checkmark   &     \checkmark  &     \checkmark  & \checkmark      \\
		\hline		
	\end{tabular} 
    \label{table:ex}
\end{table}

\subsection{Contributions}

In this paper, we develop a new fully-distributed algorithm for solving learning problems when the data is distributed among agents in feature partitions and computing the conjugate of the possibly non-smooth cost or regularizer functions is challenging or unfeasible. We consider a general regularized non-smooth learning problem whose cost function cannot be written as the sum of local agent-specific cost functions, i.e., it is not separable as in \eqref{eqn:canon_objective} ahead.

To tackle the problem, we articulate the associated dual optimization problem and utilize the alternating direction method of multipliers (ADMM) to solve it as, unlike the original problem, its structure is suitable for distributed treatment via the ADMM. We then consider the dual of the optimization problem associated with the ADMM update step and solve it via the block coordinate-descent (BCD) algorithm. In that manner, we devise an approach that enables us to avoid the explicit computation of any conjugate function, which may be hard or infeasible for some objective functions. The proposed algorithm is fully distributed, i.e., it only relies on single-hop communications among neighboring agents and does not need any central coordinator or processing hub. We demonstrate that the proposed algorithm approaches the optimal centralized solution at all agents. Our experiments show that the proposed algorithm converges to the optimal solution in various scenarios and is competitive with the relevant existing algorithms even when dealing with problems that, unlike its contenders, it is not tailored for.

In Table~\ref{table:ex}, we provide a comparative summary of the proposed algorithm with respect to the most relevant existing ones in terms of the key features of being fully distributed, ability to handle non-smooth cost or regularization functions, and non-reliance on any conjugate function.

\subsection{Paper Organization}

The rest of the paper is organized as follows. In Section II, we describe the system model and formulate the distributed learning problem with feature partitioning when both the cost and regularizer functions are convex but not necessarily smooth. In Section III, we describe our proposed algorithm for solving the considered regularized learning problem in a distributed fashion without computing any conjugate function. Subsequently, we prove the convergence of the proposed algorithm by confirming that both its inner and outer loops converge in Section IV. We provide some simulation results in Section V and draw conclusions in Section VI.

\subsection{Mathematical Notations}

The set of natural and real numbers are denoted by $\mathbb{N}$ and $\mathbb{R}$, respectively. The set of positive real numbers is denoted by $\mathbb{R}_+$. Scalars, column vectors, and matrices are respectively denoted by lowercase, bold lowercase, and bold uppercase letters. The operators $(\cdot)^\mathsf{T}$, $\text{det}(\cdot)$, and $\text{tr}(\cdot)$ denote transpose, determinant, and trace of a matrix, respectively. The symbol $\norm{\cdot}$ represents the Euclidean norm of its vector argument and $\otimes$ stands for the Kronecker product. $\mathbf{I}_n$ is an identity matrix of size $n$, $\mathbf{0}_n$ is an $n\times 1$ vector with all zeros entries, $\mathbf{0}_{n\times p}=\mathbf{0}_n\mathbf{0}_p^\mathsf{T}$, and $|\cdot|$ denotes the cardinality operator if its argument is a set. The statistical expectation and covariance operators are represented by $\mathbb{E}[\cdot]$ and $\text{cov}[\cdot]$, respectively. For any positive semidefinite matrix $\mathbf{X}$, $\lambda_{\min}(\mathbf{X})$ and $\lambda_{\max}(\mathbf{X})$ denote the nonzero smallest and largest eigenvalues of $\mathbf{X}$, respectively. For a vector $\mathbf{x}\in\mathbb{R}^n$ and a positive semi-definite matrix $\mathbf{A}$, $\norm{\mathbf{x}}^2_{\mathbf{A}}$ denotes the quadratic form $\mathbf{x}^\mathsf{T}\mathbf{A}\mathbf{x}$. The conjugate function of any function $f$ is denoted by $f^*$. 

\begin{figure}[t!]
\begin{center}
\includegraphics[scale=0.45]{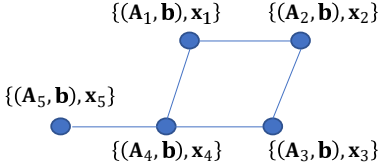}
\caption{Distributed features over a network with five agents.}
\label{second}
\end{center}
\end{figure}

\section{System Model}

We model a network with $N\in\mathbb{N}$ agents and $E\in\mathbb{N}$ edges as an undirected graph $\mathcal{G}(\mathcal{V},\mathcal{E})$ with the set of vertices $\mathcal{V}=\{1,\dots,N\}$ corresponding to the agents and the set of edges $\mathcal{E}$ standing for the bidirectional communication links between the pairs of agents. Agent $i\in\mathcal{V}$ communicates only with its neighbors specified by the set $\mathcal{V}_i$.

Let us denote the network-wide data as an observation matrix $\mathbf{A}\in\mathbb{R}^{M\times P}$ and a response vector $\mathbf{b}\in\mathbb{R}^{M\times 1}$ where $M$ is the number of data samples and $P$ is the total number of features across the network. As we consider the feature partitioning of the data, we denote the observation matrix of the $i$th agent by $\mathbf{A}_i\in\mathbb{R}^{M\times P_i}$ and its local model vector by $\mathbf{x}_i\in\mathbb{R}^{P_i\times 1}$ where $P_i$ is the number of features specific to agent $i$. Accordingly, we have $P=\sum_{i=1}^NP_i$ and $\mathbf{A}$ consists of $N$ submatrices $\mathbf{A}_i$ as
$$\mathbf{A}=[\mathbf{A}_1,\mathbf{A}_2,\hdots,\mathbf{A}_N].$$
The network-wide model vector $\mathbf{x}\in\mathbb{R}^{P\times 1}$ that relates $\mathbf{A}$ and $\mathbf{b}$ is also a stack of $N$ subvectors $\mathbf{x}_i$ as
$$\mathbf{x}=\left[ \mathbf{x}_1^{\text{T}},\mathbf{x}_2^{\text{T}},\hdots,\mathbf{x}_N^{\text{T}} \right]^{\text{T}}.$$
We give an example for feature partitioning of data in Fig.~\ref{second} where features are distributed over a network of five agents.

We consider a regularized learning problem consisting in minimizing a global cost function $f(\cdot)$ that is a function of the error $\mathbf{A}\mathbf{x}-\mathbf{b}$ and is added by a regularization function $r(\cdot)$. In the centralized approach, the optimal solution is given by
\begin{equation}
   \mathbf{x}^o=\arg\min_{\mathbf{x}}\Bigl\{ f \left(\mathbf{A}\mathbf{x}-\mathbf{b}\right)+r(\mathbf{x})\Bigr\}.
	\label{eqn:canon_objective_v0}
\end{equation}

Considering feature partitioning of the data, $\mathbf{A}\mathbf{x}$ can be written as $$\mathbf{A}\mathbf{x}=\sum_{i=1}^N\mathbf{A}_i\mathbf{x}_i$$ and assuming that the regularizer function $r(\cdot)$ can be written as a sum of agent-specific regularizer functions as 
\begin{equation*}
    r(\mathbf{x})=\sum_{i=1}^{N}r_i(\mathbf{x}_i),
\end{equation*}
the regularized learning problem \eqref{eqn:canon_objective_v0} is of the following form
\begin{equation}
   \underset{\{\mathbf{x}_i\} }{\text{min}}  \quad  f \left({\textstyle\sum_{i=1}^{N}} \mathbf{A}_i\mathbf{x}_i-\mathbf{b}\right) +\sum_{i=1}^{N}r_i(\mathbf{x}_i).
	\label{eqn:canon_objective}
\end{equation}

The learning problem \eqref{eqn:canon_objective} pertains to several applications in machine learning, e.g., regression over distributed features \cite{Boyd2010}, clustering in graphs \cite{Boyd_clustering_graphs}, smart grid control \cite{Scaglione}, dictionary learning \cite{Sayed_dictionary}, and network utility maximization \cite{Palomar_utility}. Similar to most existing works, e.g., \cite{Boyd2010,Ying2019,Fukushima1992}, we consider learning problems where functions $f(\cdot)$ and $r_i(\cdot)$,  $i=1,\hdots,N$, are convex, proper, and lower semi-continuous. However, in this work, the objective functions are not necessarily smooth or their conjugate functions known. Therefore, we propose a novel algorithm that solves \eqref{eqn:canon_objective} in a fully distributed fashion wherein each agent communicates only with its neighbors without requiring the computation of any conjugate function. In the next section, we describe our proposed algorithm.

\section{Algorithm}

We first present the reformulation of the considered non-separable problem into a dual form that can be solved in a fully-distributed fashion via the ADMM. Subsequently, we describe a new approach to perform the ADMM primal update step without explicitly computing any conjugate function of the cost or regularizer functions.

\subsection{Distributed ADMM for the Dual Problem}

To develop a distributed solution, we introduce the auxiliary variables $\{\mathbf{z}_i\}_{i=1}^N$ and recast \eqref{eqn:canon_objective} as
\begin{equation}
\begin{aligned}
& \underset{\{\mathbf{x}_i,\mathbf{z}_i\} }{\text{min}} 
& &  f \left({\textstyle\sum}_{i=1}^{N} \mathbf{z}_i-\mathbf{b}\right)+ \sum_{i=1}^{N}r_i(\mathbf{x}_i)     \\
&  \quad \text{s. t.  } 
& &  \mathbf{A}_i\mathbf{x}_i = \mathbf{z}_i , \quad i=1,\hdots,N.
\end{aligned}
\label{eqn:main_opt_aux}
\end{equation} 
The cost function $f(\cdot)$ in \eqref{eqn:main_opt_aux} is not separable among the agents. We consider the dual problem of \eqref{eqn:main_opt_aux} and exploit its separability property, which is lacking in the primal domain, to solve it by employing the ADMM. For this purpose, we associate the Lagrange multipliers $\{\boldsymbol{\mu}_i\}_{i=1}^N$ with the equality constraints in \eqref{eqn:main_opt_aux} and state the related Lagrangian function as
\begin{equation}
\begin{aligned}
& \mathcal{L}(\{\mathbf{x}_i\},\{\mathbf{z}_i\},\{\boldsymbol{\mu}_i\})\\
=&f \left({\textstyle\sum}_{i=1}^{N} \mathbf{z}_i-\mathbf{b}\right)+ \sum_{i=1}^{N}r_i(\mathbf{x}_i)+\sum_{i=1}^N\boldsymbol{\mu}_i^\mathsf{T}(\mathbf{A}_i\mathbf{x}_i-\mathbf{z}_i)\\
=&\sum_{i=1}^N\left(r_i(\mathbf{x}_i)+(\mathbf{A}_i^\mathsf{T}\boldsymbol{\mu}_i)^\mathsf{T}\mathbf{x}_i\right)\\
&+f\left(\sum_{i=1}^N\mathbf{z}_i-\mathbf{b}\right)-\sum_{i=1}^N\boldsymbol{\mu}_i^\mathsf{T}\mathbf{z}_i.
\end{aligned}
\label{eqn:main_lagr}
\end{equation} 

The dual function for problem \eqref{eqn:main_opt_aux} can be computed as 
\begin{equation}
\begin{aligned}
&d(\{\boldsymbol{\mu}_i\})=\inf_{\{\mathbf{x}_i,\mathbf{z}_i\}}\mathcal{L}(\{\mathbf{x}_i\},\{\mathbf{z}_i\},\{\boldsymbol{\mu}_i\}) \\
=& -\sum_{i=1}^Nr_i^*(-\mathbf{A}_i^\mathsf{T}\boldsymbol{\mu}_i)+\inf_{\mathbf{z}_i}f(\sum_{i=1}^N\mathbf{z}_i-\mathbf{b})-\sum_{i=1}^N\boldsymbol{\mu}_i^\mathsf{T}\mathbf{z}_i
	\label{eqn:dual_func_main}
	\end{aligned}
\end{equation}
where $r_i^*$ is the conjugate function of $r$ defined as
$$r_i^*(\mathbf{y}) = \sup_{\mathbf{x}}\; \mathbf{y}^T\mathbf{x} - r_i(\mathbf{x}).$$
Introducing auxiliary variable $\mathbf{z}$ that is defined as  $$\mathbf{z}=\sum_{i=1}^N\mathbf{z}_i$$ and using the duality theory, an alternate form of the dual function \eqref{eqn:dual_func_main} is given by
\begin{equation}
\tilde{d}(\{\boldsymbol{\mu}_i\},\boldsymbol{\lambda})= -f^* (\boldsymbol{\lambda})-\boldsymbol{\lambda}^\mathsf{T}\mathbf{b}- \sum_{i=1}^{N}r_i^*(-\mathbf{A}_i^\mathsf{T} \boldsymbol{\mu}_i)
\label{eqn:dual_func_main_1}
\end{equation}
when $\boldsymbol{\lambda} = \boldsymbol{\mu}_i$ $\forall i\in\mathcal{V}$ with $\boldsymbol{\lambda}$ being the dual variable corresponding to $\mathbf{z}=\sum_{i=1}^N\mathbf{z}_i$. Otherwise, we have $\tilde{d}(\{\boldsymbol{\mu}_i\},\boldsymbol{\lambda})=-\infty$.

By eliminating $\boldsymbol{\lambda}$, the dual problem for \eqref{eqn:main_opt_aux} can be expressed as 
\begin{equation}
\begin{aligned}
& \underset{\{\boldsymbol{\mu}_i\} }{\text{min}} 
& &  \frac{1}{N}\sum_{i=1}^{N}\left(f^* (\boldsymbol{\mu}_i)+\boldsymbol{\mu}_i^\mathsf{T}\mathbf{b}\right) + \sum_{i=1}^{N}r_i^*(-\mathbf{A}_i^\mathsf{T} \boldsymbol{\mu}_i)     \\
&  \quad \text{s. t.  } 
& &   \boldsymbol{\mu}_1=\boldsymbol{\mu}_2 = \cdots = \boldsymbol{\mu}_N.
\end{aligned}
\label{eqn:dual_consensus_form}
\end{equation}
To facilitate a fully-distributed solution, we decouple the constraints in \eqref{eqn:dual_consensus_form} as
\begin{equation}
\boldsymbol{\mu}_i=\mathbf{u}_i^j,\quad  \boldsymbol{\mu}_j=\mathbf{u}_i^j, \ j\in\mathcal{V}_i, \  i=1,\hdots,N
\label{eqn:main_opt_aux_re}
\end{equation}
where $\{\mathbf{u}_{i}^{j}\}_{i\in\mathcal{V},j\in\mathcal{V}_i}$ are auxiliary variables that will eventually be eliminated. 
We generate a new augmented Lagrangian function by associating the new Lagrange multipliers $\{\bar{\mathbf{v}}_i^j\}_{j\in\mathcal{V}_i}$ and $\{\tilde{\mathbf{v}}_i^j\}_{j\in\mathcal{V}_i}$ with the consensus constraints in \eqref{eqn:main_opt_aux_re}. By using the Karush-Kuhn-Tucker conditions of optimality for \eqref{eqn:main_opt_aux_re} and setting $$\mathbf{v}_i^{(k)}=2\sum_{j\in\mathcal{V}_i}(\bar{\mathbf{v}}_i^j)^{(k)},$$ it can be shown that the Lagrange multipliers $\{\tilde{\mathbf{v}}_i^j\}_{j\in\mathcal{V}_i}$ and the auxiliary variables $\{\mathbf{u}_i^j\}_{j\in\mathcal{V}_i}$ are eliminated \cite{Giannakis2016}. Hence, the ADMM to solve \eqref{eqn:dual_consensus_form} reduces to the following iterative updates at the $i$th agent  

\begin{align}
 \boldsymbol{\mu}_i^{(k)} &= \arg \min_{\boldsymbol{\mu}_i}   \;\Bigl\{\frac{1}{N}f^* (\boldsymbol{\mu}_i)+\frac{1}{N}\boldsymbol{\mu}_i^\mathsf{T}\mathbf{b}+r_i^*(-\mathbf{A}_i^\mathsf{T} \boldsymbol{\mu}_i)  \nonumber\\
&+\boldsymbol{\mu}_i^\mathsf{T}\mathbf{v}_i^{(k-1)}+ \rho \sum_{j \in \mathcal{V}_i}\Big\|\boldsymbol{\mu}_i - \frac{\boldsymbol{\mu}_i^{(k-1)}+\boldsymbol{\mu}_j^{(k-1)}}{2}\Big\|^2\Bigr\} \label{eqn:ADMM_dual_org_1}\\
 \mathbf{v}_i^{(k)}&=\mathbf{v}_i^{(k-1)} + \rho \sum_{j \in \mathcal{V}_i} (\boldsymbol{\mu}_i^{(k)}-\boldsymbol{\mu}_j^{(k)})\label{eqn:ADMM_dual_org_2}
\end{align}
where $\rho>0$ is the penalty parameter and $k$ is the iteration index. 

The objective function in \eqref{eqn:ADMM_dual_org_1} may be non-smooth as the global cost function $f(\cdot)$ or the agent-specific regularizer functions $r_i(\cdot)$, $i=1,\hdots,N$, and consequently their conjugate functions may be non-smooth. Thus, the minimization problem in \eqref{eqn:ADMM_dual_org_1} can be solved by employing suitable subgradient methods or proximal operators \cite{Bertsekas99,Boyd2014}. 
However, computing the conjugate functions of the cost or the regularizer functions in \eqref{eqn:ADMM_dual_org_1} may be hard or even unfeasible. To overcome this challenge, in the next subsection, we describe a new approach that does not require the explicit calculation of any conjugate function. 

\subsection{ADMM without Conjugate Function}

We rewrite the minimization problem in the ADMM primal update  \eqref{eqn:ADMM_dual_org_1} as 
\begin{equation}
\begin{aligned}
\boldsymbol{\mu}_i^{(k)}=\arg\min_{\{\boldsymbol{\mu}_i,\boldsymbol{\nu}_i,\boldsymbol{\alpha}_i\}} 
&\Bigl\{ \frac{f^*(\boldsymbol{\mu}_i)+\boldsymbol{\mu}_i^\mathsf{T}\mathbf{b}}{N}+r_i^*(\boldsymbol{\nu}_i)\\&+\boldsymbol{\mu}_i^\mathsf{T}\mathbf{c}_i^{(k-1)}+\bar{\rho}_i\norm{\boldsymbol{\alpha}_i}^2\Bigr\}\\
\text{s.t.\quad\ } 
&\ \mathbf{A}_i^\mathsf{T}\boldsymbol{\mu}_i+\boldsymbol{\nu}_i=\mathbf{0}\\
&\ \boldsymbol{\mu}_i=\boldsymbol{\alpha}_i
\end{aligned}
\label{first_sect_b}
\end{equation}
where $\bar{\rho}_i=\rho|\mathcal{V}_i|$ and
\begin{equation}
\label{second_sect_b}
\mathbf{c}_i^{(k-1)}=\mathbf{v}_i^{(k-1)}-\rho|\mathcal{V}_i|\boldsymbol{\mu}_i^{(k-1)}-\rho\sum_{j\in\mathcal{V}_i}\boldsymbol{\mu}_j^{(k-1)}.
\end{equation}
The Lagrangian function related to \eqref{first_sect_b} is stated as
\begin{equation}
\begin{aligned}
\mathcal{L}_k(\boldsymbol{\mu}_i,\boldsymbol{\nu}_i,\boldsymbol{\alpha}_i,\boldsymbol{\theta}_i^{(k)},\boldsymbol{\beta}_i^{(k)})
&=\frac{f^*(\boldsymbol{\mu}_i)+\boldsymbol{\mu}_i^\mathsf{T}\mathbf{b}}{N}+r_i^*(\boldsymbol{\nu}_i)+\boldsymbol{\mu}_i^\mathsf{T}\mathbf{c}_i^{(k-1)}\\
&+\bar{\rho}_i\norm{\boldsymbol{\alpha}_i}^2+(\boldsymbol{\theta}_i^{(k)})^\mathsf{T}(\mathbf{A}_i^\mathsf{T}\boldsymbol{\mu}_i+\boldsymbol{\nu}_i)\\
&+(\boldsymbol{\beta}_i^{(k)})^\mathsf{T}(\boldsymbol{\mu}_i-\boldsymbol{\alpha}_i)
\end{aligned}
\label{third_sect_b}
\end{equation}
where $\boldsymbol{\theta}_i^{(k)}$ and $\boldsymbol{\beta}_i^{(k)}$ are the Lagrange multipliers associated with the first and the second constraints in \eqref{first_sect_b}, respectively, at iteration $k$.

Motivated by the close connection between a function and its double conjugate (conjugate of conjugate), we express the dual for the objective in \eqref{first_sect_b} as
\begin{equation}
\begin{aligned}
&\delta_k(\boldsymbol{\theta}_i^{(k)},\boldsymbol{\beta}_i^{(k)})
=\inf_{\{\boldsymbol{\mu}_i,\boldsymbol{\nu}_i,\boldsymbol{\alpha}_i\}}\mathcal{L}_k(\boldsymbol{\mu}_i,\boldsymbol{\nu}_i,\boldsymbol{\alpha}_i,\boldsymbol{\theta}_i^{(k)},\boldsymbol{\beta}_i^{(k)})\\
=&\inf_{\boldsymbol{\nu}_i}\{r_i^*(\boldsymbol{\nu}_i)+(\boldsymbol{\theta}_i^{(k)})^\mathsf{T}\boldsymbol{\nu}_i\}+\inf_{\boldsymbol{\alpha}_i}\{\bar{\rho}_i\norm{\boldsymbol{\alpha}_i}^2-(\boldsymbol{\beta}_i^{(k)})^\mathsf{T}\boldsymbol{\alpha}_i\}\\
+&\inf_{\boldsymbol{\mu}_i}\Bigl\{\frac{f^*(\boldsymbol{\mu}_i)}{N}+\left(\mathbf{c}_i^{(k-1)}+\mathbf{A}_i\boldsymbol{\theta}_i^{(k)}+\frac{\mathbf{b}}{N}+\boldsymbol{\beta}_i^{(k)}\right)^\mathsf{T}\boldsymbol{\mu}_i\Bigr\}.
\end{aligned}
\label{fourth_sect_b}
\end{equation}
By employing the definition of conjugate function, the first infimum in \eqref{fourth_sect_b} is equal to $-r_i^{**}(-\boldsymbol{\theta}_i^{(k)})$. The second infimum in \eqref{fourth_sect_b} can be easily obtained by noting that the function $$l_k(\boldsymbol{\alpha}_i)\coloneqq \bar{\rho}_i\norm{\boldsymbol{\alpha}_i}^2-(\boldsymbol{\beta}_i^{(k)})^\mathsf{T}\boldsymbol{\alpha}_i$$ is quadratic in $\boldsymbol{\alpha}_i$. Hence, this infimum can be calculated by computing the gradient of $l_k(\cdot)$ and equating it to zero, i.e., $$\bar{\rho}_i\norm{\boldsymbol{\alpha}_i}^2-(\boldsymbol{\beta}_i^{(k)})^\mathsf{T}\boldsymbol{\alpha}_i=0.$$ Solving this equation for $\boldsymbol{\alpha}_i$ gives
\begin{equation}
\label{fifth_sect_b}
\boldsymbol{\alpha}_i^o=\frac{\boldsymbol{\beta}_i^{(k)}}{2\bar{\rho}_i}.
\end{equation}
This implies that the second infimum in \eqref{fourth_sect_b} is attained at the optimal value $\boldsymbol{\alpha}_i^o$, which in turn means that it is equal to
$$l_k(\boldsymbol{\alpha}_i^o)=-\frac{\norm{\boldsymbol{\beta}_i^{(k)}}^2}{4\bar{\rho}_i}.$$

In view of the definition and properties of the conjugate function \cite{BoydStephenP2004Co}, the third infimum in \eqref{fourth_sect_b} is given by
$$-Nf^{**}\left(-\mathbf{c}_i^{(k-1)}-\mathbf{A}_i\boldsymbol{\theta}_i^{(k)}-\frac{\mathbf{b}}{N}-\boldsymbol{\beta}_i^{(k)}\right).$$
Therefore, we have
\begin{equation}
\begin{aligned}
\delta_k(\boldsymbol{\theta}_i^{(k)},\boldsymbol{\beta}_i^{(k)})=&-r_i^{**}(-\boldsymbol{\theta}_i^{(k)})-\frac{\norm{\boldsymbol{\beta}_i^{(k)}}^2}{4\bar{\rho}_i}\\
&-Nf^{**}\left(-\mathbf{c}_i^{(k-1)}-\mathbf{A}_i\boldsymbol{\theta}_i^{(k)}-\frac{\mathbf{b}}{N}-\boldsymbol{\beta}_i^{(k)}\right).
\end{aligned}
\label{sixth_sect_b}
\end{equation}
Since $f(\cdot)$ and $r_i(\cdot)$ are convex, proper, and lower semi-continuous, we know $f^{**}=f$ and $r_i^{**}=r_i$ due to the Fenchel Moreau Theorem \cite{Borwein}. Therefore, we have
\begin{equation}
\begin{aligned}
\delta_k(\boldsymbol{\theta}_i^{(k)},\boldsymbol{\beta}_i^{(k)})=&-r_i(-\boldsymbol{\theta}_i^{(k)})-\frac{\norm{\boldsymbol{\beta}_i^{(k)}}^2}{4\bar{\rho}_i}\\
&-Nf\left(-\mathbf{c}_i^{(k-1)}-\mathbf{A}_i\boldsymbol{\theta}_i^{(k)}-\frac{\mathbf{b}}{N}-\boldsymbol{\beta}_i^{(k)}\right).
\end{aligned}
\label{seventh_sect_b}
\end{equation}

To find the optimal $(\boldsymbol{\theta}_i^{(k)},\boldsymbol{\beta}_i^{(k)})$, we need to maximize $\delta_k(\boldsymbol{\theta}_i^{(k)},\boldsymbol{\beta}_i^{(k)})$ or, equivalently, to minimize $-\delta_k(\boldsymbol{\theta}_i^{(k)},\boldsymbol{\beta}_i^{(k)})$. Since this is a function of two variables $\boldsymbol{\theta}_i^{(k)}$ and $\boldsymbol{\beta}_i^{(k)}$, we employ the block coordinate descent algorithm (BCD) to minimize $-\delta_k(\boldsymbol{\theta}_i^{(k)},\boldsymbol{\beta}_i^{(k)})$ and find the optimal values for $(\boldsymbol{\theta}_i^{(k)},\boldsymbol{\beta}_i^{(k)})$. The BCD steps are obtained by alternatively minimizing $-\delta_k(\boldsymbol{\theta}_i^{(k)},\boldsymbol{\beta}_i^{(k)})$ with respect to $\boldsymbol{\theta}_i^{(k)}$ and $\boldsymbol{\beta}_i^{(k)}$ as follows
 \begin{align}
\boldsymbol{\theta}_i^{(k,t)}
&=\arg\min_{\boldsymbol{\theta}_i^{(k)}}\Bigl\{r_i(-\boldsymbol{\theta}_i^{(k)})\nonumber\\
&\quad\quad+Nf\left(-\mathbf{c}_i^{(k-1)}-\mathbf{A}_i\boldsymbol{\theta}_i^{(k,t)}-\frac{\mathbf{b}}{N}-\boldsymbol{\beta}_i^{(k,t-1)}\right)\Bigr\}\label{tenth_sect_b_a}\\
\boldsymbol{\beta}_i^{(k,t)}
&=\arg\min_{\boldsymbol{\beta}_i^{(k)}}\Bigl\{\frac{1}{4\bar{\rho}_i}\norm{\boldsymbol{\beta}_i^{(k)}}^2\nonumber\\
&\quad\quad+Nf\left(-\mathbf{c}_i^{(k-1)}-\mathbf{A}_i\boldsymbol{\theta}_i^{(t)}-\frac{\mathbf{b}}{N}-\boldsymbol{\beta}_i^{(k)}\right)\Bigr\}
\label{tenth_sect_b_b}
\end{align}
where $t$ is the BCD iteration index. If we assume that the BCD algorithm converges after $T$ iterations. The optimal values of $\boldsymbol{\theta}_i^{(k)}$ and $\boldsymbol{\beta}_i^{(k)}$ can be denoted by $\boldsymbol{\theta}_i^{(k,T)}$ and $\boldsymbol{\beta}_i^{(k,T)}$, respectively.

To update the Lagrange multipliers $\boldsymbol{\mu}_i^{(k)}$ we employ the complementary slackness conditions, i.e., 
$$\boldsymbol{\beta}_i^{(k,T)}(\boldsymbol{\mu}_i^{(k)}-\boldsymbol{\alpha}_i^{o})=\mathbf{0}\quad \forall i\in\mathcal{V}.$$
Since $\boldsymbol{\beta}_i^{(k,T)}\neq\mathbf{0}$, $\forall i\in\mathcal{V}$, we have $$\boldsymbol{\mu}_i^{(k)}-\boldsymbol{\alpha}_i^{o}=\mathbf{0}\quad \forall i\in\mathcal{V}.$$ Using \eqref{fifth_sect_b}, we can update $\boldsymbol{\mu}_i^{(k)}$ as 
\begin{equation}
    \label{eighth_sect_b}
    \boldsymbol{\mu}_i^{(k)}=\frac{\boldsymbol{\beta}_i^{(k,T)}}{2\bar{\rho}_i}.
\end{equation}
Collating the expressions in \eqref{eighth_sect_b}, \eqref{second_sect_b}, \eqref{eqn:ADMM_dual_org_2}, the ADMM steps in \eqref{eqn:ADMM_dual_org_1} and \eqref{eqn:ADMM_dual_org_2} can be equivalently expressed as 
\begin{align}
\boldsymbol{\mu}_i^{(k)}&=\frac{\boldsymbol{\beta}_i^{(k,T)}}{2\bar{\rho}_i}\label{ninth_sect_b_a}\\
\mathbf{v}_i^{(k)}&=\mathbf{v}_i^{(k-1)} + \rho \sum_{j \in \mathcal{V}_i} (\boldsymbol{\mu}_i^{(k)}-\boldsymbol{\mu}_j^{(k)})\label{ninth_sect_b_b}\\
\mathbf{c}_i^{(k-1)}&=\mathbf{v}_i^{(k-1)}-\rho|\mathcal{V}_i|\boldsymbol{\mu}_i^{(k-1)}-\rho\sum_{j\in\mathcal{V}_i}\boldsymbol{\mu}_j^{(k-1)}\label{ninth_sect_b_c}
\end{align}
where $k$ is the ADMM iteration index. We summarize the proposed algorithm in Algorithm 1.

\begin{algorithm}[t!]
\caption{The proposed algorithm for feature-partitioned distributed learning with unknown conjugate functions}
\label{alg:FDA}
 \begin{algorithmic}
  \STATE At all agents $i\in\mathcal{V}$, initialize $\boldsymbol{\mu}_i^{(0)}=\mathbf{0}$, $\mathbf{v}_i^{(0)}=\mathbf{0}$, and locally run:
   \FOR{$k=1,2,\hdots,K$} 
   \STATE Run BCD loop
   \FOR{$t=1,2,\hdots,T$}
   \STATE Update $\boldsymbol{\theta}_i^{(k,t)}$ via \eqref{tenth_sect_b_a}.
   \STATE Update $\boldsymbol{\beta}_i^{(k,t)}$ via \eqref{tenth_sect_b_b}.
   \ENDFOR
    \STATE  Update the dual variables $\boldsymbol{\mu}_i^{(k)}=\boldsymbol{\beta}_i^{(k,T)}/(2\bar{\rho}_i)$.
	\STATE Share $\boldsymbol{\mu}_i^{(k)}$ with the neighbors in $\mathcal{V}_i$.    
    \STATE Update the Lagrange multipliers $\mathbf{v}_i^{(k)}$ via \eqref{ninth_sect_b_b}.
	\STATE Update the auxiliary variables $\mathbf{c}_i^{(k)}$ via \eqref{ninth_sect_b_c}.
	\ENDFOR
 \end{algorithmic} 
 \end{algorithm}

Assuming that the ADMM outer loop converges after $K$ iterations, we denote the optimal dual variable $\boldsymbol{\theta}_i^{(K,T)}$ by $\boldsymbol{\theta}_i^o$. The estimate $\boldsymbol{\theta}_i^o$ at agent $i$ is indeed the optimal solution to the original problem \eqref{eqn:canon_objective}, i.e., $\mathbf{x}_i^o$, as per the following theorem.

\begin{theo}
For all agents $i\in\mathcal{V}$, the optimal dual variable $\boldsymbol{\theta}_i^o$ at agent $i$ is equal to the optimal estimate $\mathbf{x}_i^o$, i.e., $\boldsymbol{\theta}_i^o=\mathbf{x}_i^o$.
\end{theo}
\begin{proof}
Since the optimization problem in \eqref{eqn:main_opt_aux} has a convex objective and is feasible, the Slater's condition is satisfied. Therefore, due to the Slater's theorem, strong duality holds and $\boldsymbol{\theta}_i^o=\mathbf{x}_i^o$, $\forall i\in\mathcal{V}$ \cite{BoydStephenP2004Co}.
\end{proof}

\section{Convergence Analysis}\label{conv-anal}

The convergence of the proposed algorithm can be proven by corroborating that both the inner-loop BCD and outer-loop ADMM iterations converge. First, the convergence of the inner loop can be verified from results in \cite{Mingyihongbook} since all the assumptions required for the convergence are satisfied, i.e., the function $\delta(\cdot)$ is convex and the feasible sets $\mathbb{R}^M$ and $\mathbb{R}^{P_i}$, $\forall i\in\mathcal{V}$, are all convex. Assuming that the optimal solution $\boldsymbol{\beta}_i^o$ of the inner-loop BCD algorithm is attained for each $i\in\mathcal{V}$, the dual variable $\boldsymbol{\mu}_i^{(k)}$ in the outer loop is updated accordingly. 

Next, we prove that the estimates produced by the fully-distributed ADMM outer loop, i.e., \eqref{eqn:ADMM_dual_org_1} and \eqref{eqn:ADMM_dual_org_2}, approach the optimal centralized solution at all agents. To present the convergence result, we rewrite the constraints in \eqref{eqn:dual_consensus_form} as follows
\begin{equation}
\begin{aligned}
\boldsymbol{\mu}_i=\bar{\mathbf{u}}_i^j,\quad \boldsymbol{\mu}_j=\Breve{\mathbf{u}}_i^j, \quad \bar{\mathbf{u}}_i^j=\Breve{\mathbf{u}}_i^j, \quad j\in\mathcal{V}_i, \; i=1,\hdots,N.
\end{aligned}
\label{eqn:main_opt_aux_re_conv_proof}
\end{equation}
Note that the constraints $\mathbf{u}\in\mathcal{C}_{\mathbf{u}}\coloneqq\{\mathbf{u}: \bar{\mathbf{u}}_i^j=\Breve{\mathbf{u}}_i^j,\ i\in\mathcal{V},\ j\in\mathcal{N}_i\}$ are not dualized and are introduced only to present the convergence result. Let us define the following vectors
\begin{equation*}
    \begin{aligned}
    \boldsymbol{\mu}=&[\boldsymbol{\mu}_1^\mathsf{T},\hdots,\boldsymbol{\mu}_N^\mathsf{T}]^\mathsf{T}\\
    \mathbf{u}=&[(\mathbf{u}_1^{a_1(1)})^\mathsf{T},\hdots,(\mathbf{u}_1^{a_N(|\mathcal{V}_N|)})^\mathsf{T},\\
    &\ \ \ \hdots,(\mathbf{u}_N^{a_N(1)})^\mathsf{T},\hdots,(\mathbf{u}_N^{a_N(|\mathcal{V}_N|)})]^\mathsf{T}
    \end{aligned}
\end{equation*}
where $a_i(j)$ is the index of the $j$th neighbor of agent $i$.

The problem \eqref{eqn:main_opt_aux_re} with the constraints in \eqref{eqn:main_opt_aux_re_conv_proof} can be written as  
\begin{equation}
\begin{aligned}
&\underset{\boldsymbol{\mu},\mathbf{u}}{\min}
&& G_1(\boldsymbol{\mu})+G_2(\mathbf{u}) \\
&\text{\ s.t.}\ 
&& \boldsymbol{\mu}\in\mathcal{C}_1,\ \mathbf{u}\in\mathcal{C}_2,\ \mathbf{C}\boldsymbol{\mu}=\mathbf{u}
\end{aligned}
\label{convergence3}
\end{equation}
where $\mathbf{C}=[\mathbf{C}_1^\mathsf{T},\mathbf{C}_2^\mathsf{T}]^\mathsf{T}$, $G_2(\mathbf{u})=0$, $\mathcal{C}_1\coloneqq\mathbb{R}^M$, $\mathcal{C}_2\coloneqq\mathcal{C}_{\mathbf{u}}$,
\begin{equation*}
\begin{aligned}
\mathbf{C}_1=&\begin{bmatrix} \mathbf{C}_{11}\\ \vdots \\ \mathbf{C}_{1N} \end{bmatrix}, \quad \mathbf{C}_{1i}\coloneqq(\mathbf{1}_{|\mathcal{N}_i|}\mathbf{e}_i^\mathsf{T})\otimes\mathbf{I}_M, \ i\in\mathcal{V}\\
\mathbf{C}_2=&\begin{bmatrix} \mathbf{C}_{21}\\ \vdots \\ \mathbf{C}_{2N} \end{bmatrix}, \quad \mathbf{C}_{2i}\coloneqq
\begin{bmatrix}
\mathbf{e}_{i_i(1)}^\mathsf{T}\\ \vdots \\ \mathbf{e}_{i_i(|\mathcal{N}_i|)}^\mathsf{T}
\end{bmatrix}\otimes\mathbf{I}_M, \ i\in\mathcal{V}\\
G_1(\boldsymbol{\mu})=&\frac{1}{N}\sum_{i=1}^N\left(f^*(\boldsymbol{\mu}_i)+\boldsymbol{\mu}_i^\mathsf{T}\mathbf{b}\right)+\sum_{i=1}^Nr_i^*(-\mathbf{A}_i^\mathsf{T}\boldsymbol{\mu}_i),
\end{aligned}
\end{equation*}
and $\mathbf{e}_i$ is the $i$th vector of the canonical basis of $\mathbb{R}^M$. The convergence result relies on the following lemma.

\begin{lem}
If $\mathcal{G}$ is a connected graph, then the local optimal solution $\boldsymbol{\mu}_i^o$ at agent $i$ is equal to the optimal centralized solution of \eqref{eqn:main_opt_aux_re}, i.e., $\boldsymbol{\mu}_i^o=\boldsymbol{\mu}^o$, $\forall i\in\mathcal{V}$, where $$\boldsymbol{\mu}^o=\arg\min_{\boldsymbol{\mu}}\Bigl\{ f^*(\boldsymbol{\mu})+\boldsymbol{\mu}^\mathsf{T}\mathbf{b} + \sum_{i=1}^{N}r_i^*(-\mathbf{A}_i^\mathsf{T} \boldsymbol{\mu})\Bigr\}.$$
\end{lem}

\begin{proof}
Let $i$ and $i'$ be arbitrary agents in $\mathcal{G}$ and $p(i,i'):i,i_1,i_2,\hdots,i_n,i'$ an arbitrary path on $\mathcal{G}$ that connects $i$ and $i'$. Since the adjacent agents in $p(i,i')$ are neighbors, we have $$\boldsymbol{\mu}_i=\boldsymbol{\mu}_{i_1}=\boldsymbol{\mu}_{i_2}=\hdots=\boldsymbol{\mu}_{i_n}=\boldsymbol{\mu}_{i'},$$ which imply $$\boldsymbol{\mu}_i=\boldsymbol{\mu}_{i'}.$$ Since $\mathcal{G}$ is connected and the path is arbitrary, the local constraints $\boldsymbol{\mu}_i=\boldsymbol{\mu}_{i'}$ can be removed and replaced by the common constraint $\boldsymbol{\mu}_i=\boldsymbol{\mu}$. Hence, $\boldsymbol{\mu}_i^o=\boldsymbol{\mu}^o$ $\forall i\in\mathcal{V}$ where $$\boldsymbol{\mu}^o=\arg\min_{\boldsymbol{\mu}}\Bigl\{ f^*(\boldsymbol{\mu})+\boldsymbol{\mu}^\mathsf{T}\mathbf{b} + \sum_{i=1}^{N}r_i^*(-\mathbf{A}_i^\mathsf{T} \boldsymbol{\mu})\Bigr\}.$$  
\end{proof}

We can now prove the convergence of the proposed algorithm as per the following theorem.

\begin{theo}
If $\mathcal{G}$ is a connected graph, then the proposed algorithm converges to the optimal centralized solution, i.e., 
\begin{equation}
    \label{first_eq_conv_proof}
    \lim_{k\rightarrow\infty}\boldsymbol{\mu}_i^{(k)}=\boldsymbol{\mu}^o, \forall i\in\mathcal{V}.
\end{equation}
\end{theo}
\begin{proof}
Thanks to Lemma 1, we only need to prove that $$\lim_{k\rightarrow\infty}\boldsymbol{\mu}_i^{(k)}=\boldsymbol{\mu}_i^o.$$ For this purpose, we observe that \eqref{convergence3} is in the same form as \cite[eq. 4.77, p. 255]{Bertsekas99}. Furthermore, the following assumptions are satisfied:
\begin{itemize}
    \item $G_1(\cdot)$ and $G_2(\cdot)$ are convex functions;
    \item $\mathcal{C}_1$ and $\mathcal{C}_2$ are nonempty polyhedral sets;
    \item $\mathbf{C}$ is full column rank, hence, $\mathbf{C}^\mathsf{T}\mathbf{C}$ is invertible.
\end{itemize}
Therefore, due to \cite[Proposition 4.2, p. 256]{Bertsekas99}, we have
$$\lim_{k\rightarrow\infty}\boldsymbol{\mu}_i^{(k)}=\boldsymbol{\mu}^o, \forall i\in\mathcal{V}.$$
\end{proof}

\section{Simulations}

\begin{figure*}[htp]
  \centering
  \subfigure[different values of $P_i$]{\includegraphics[scale=1.00]{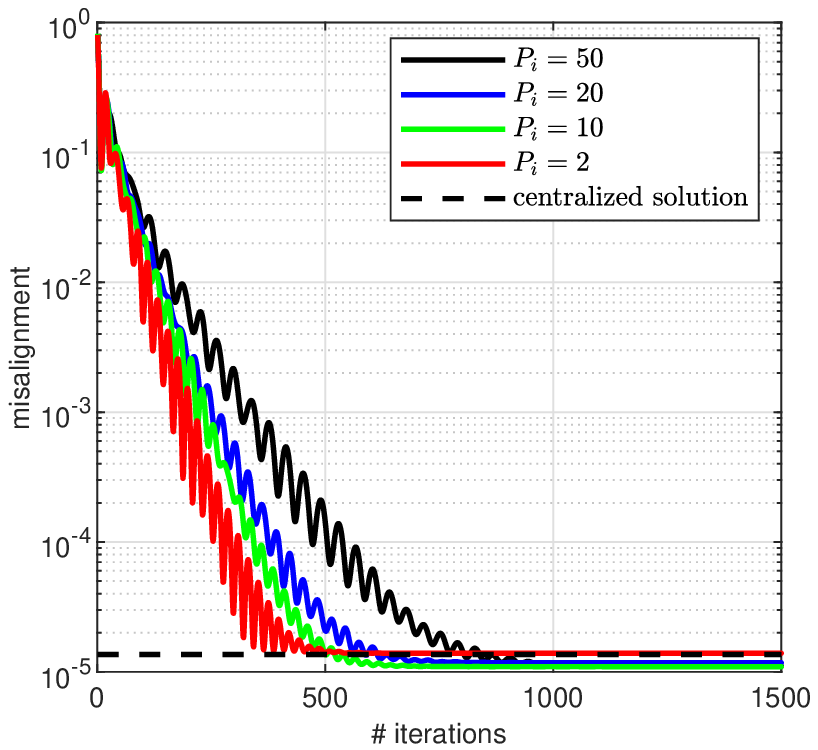}\label{fig:pi}}\quad
  \subfigure[different values of $M$]{\includegraphics[scale=1.00]{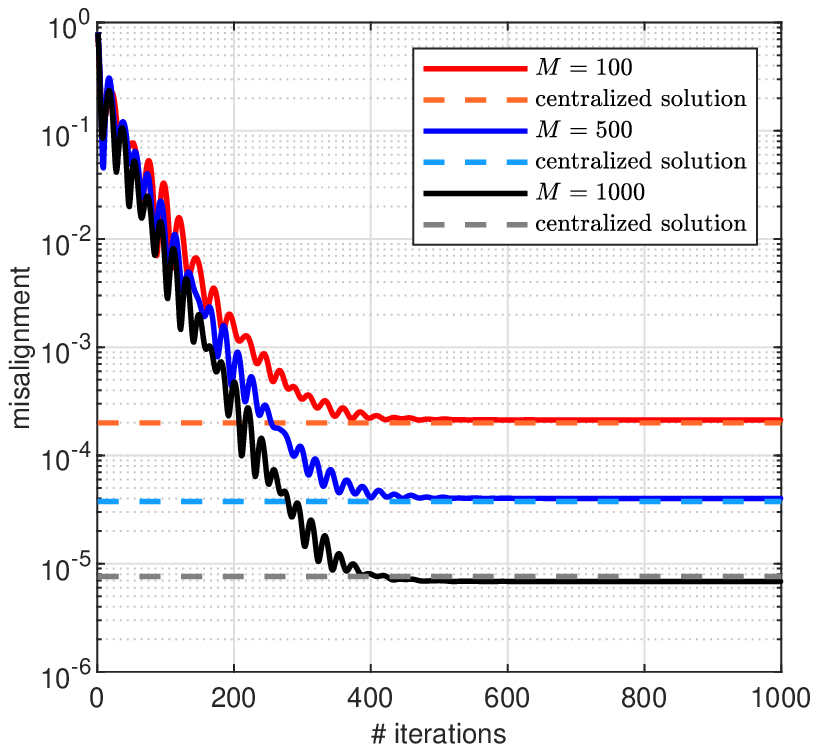}\label{fig:m}}
\caption{The misalignment of the proposed algorithm solving the distributed elastic-net regression problem with different values of $P_i$ and $M$.}  
\end{figure*}

\begin{figure*}[htp]
  \centering
  \subfigure[different values of $N$]{\includegraphics[scale=1.00]{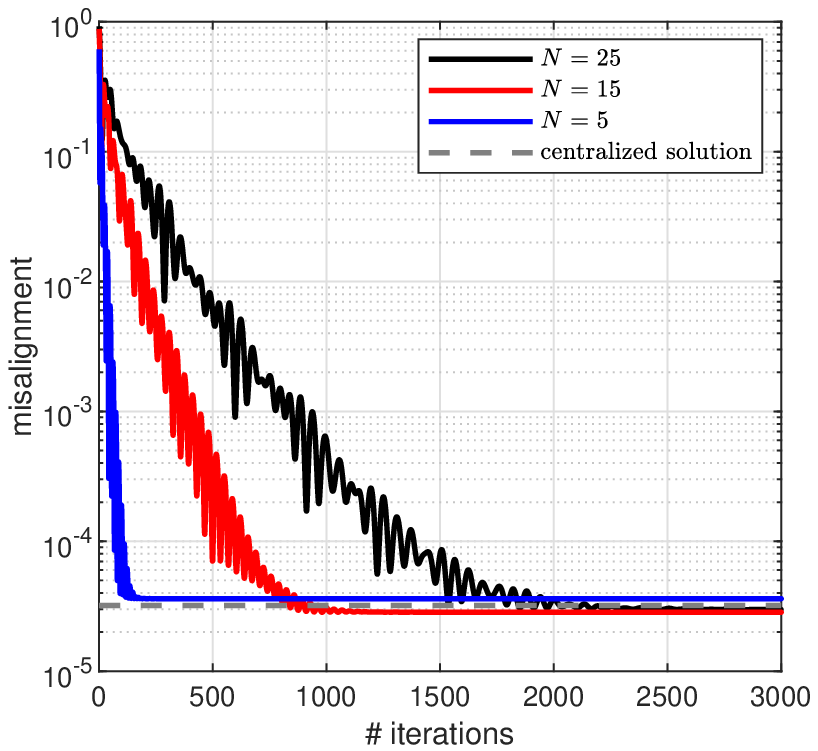}\label{fig:n}}\quad
  \subfigure[different topologies]{\includegraphics[scale=1.00]{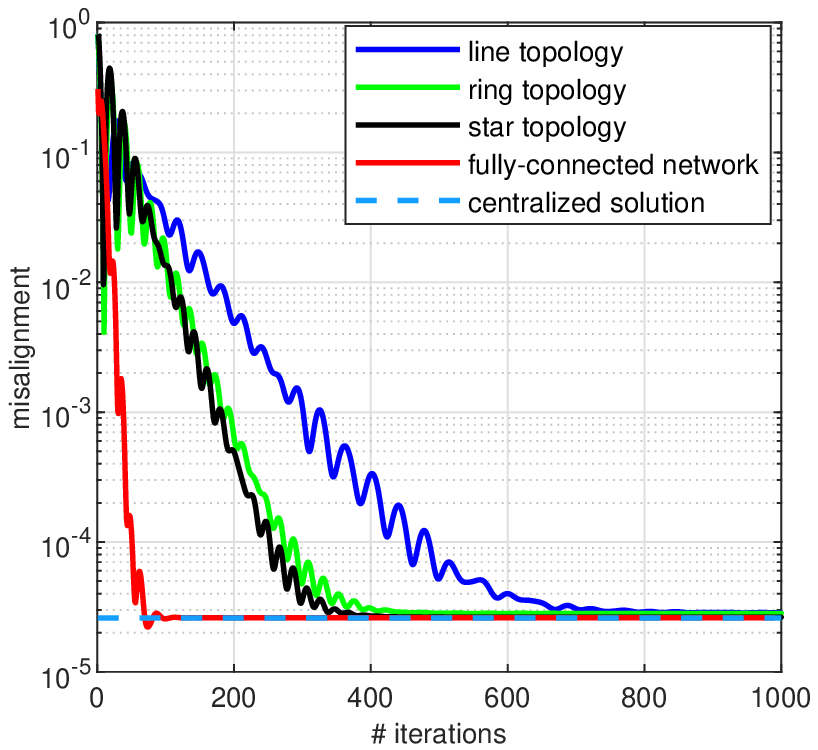}\label{fig:net}}
\caption{The misalignment of the proposed algorithm solving the distributed elastic-net regression problem with different values of $N$ and different topologies.}  
\end{figure*}

\begin{figure*}[htp]
  \centering
  \subfigure[Ridge regression with $N=10$, $M=50$, and $P_i=2$.]{\includegraphics[scale=1.00]{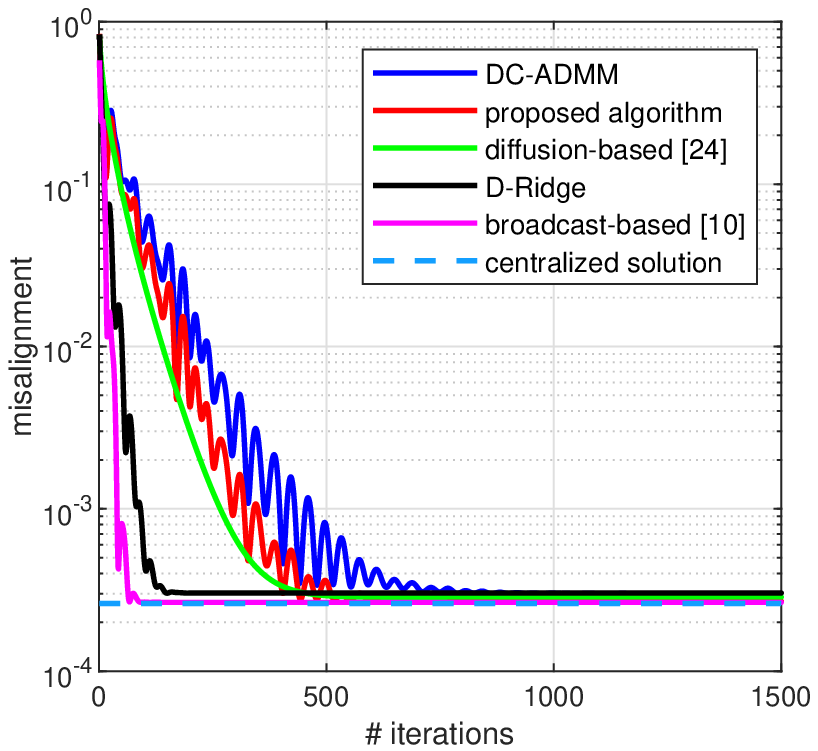}\label{fig:m50}}\quad
  \subfigure[Ridge regression with $N=10$, $M=200$, and $P_i=2$.]{\includegraphics[scale=1.00]{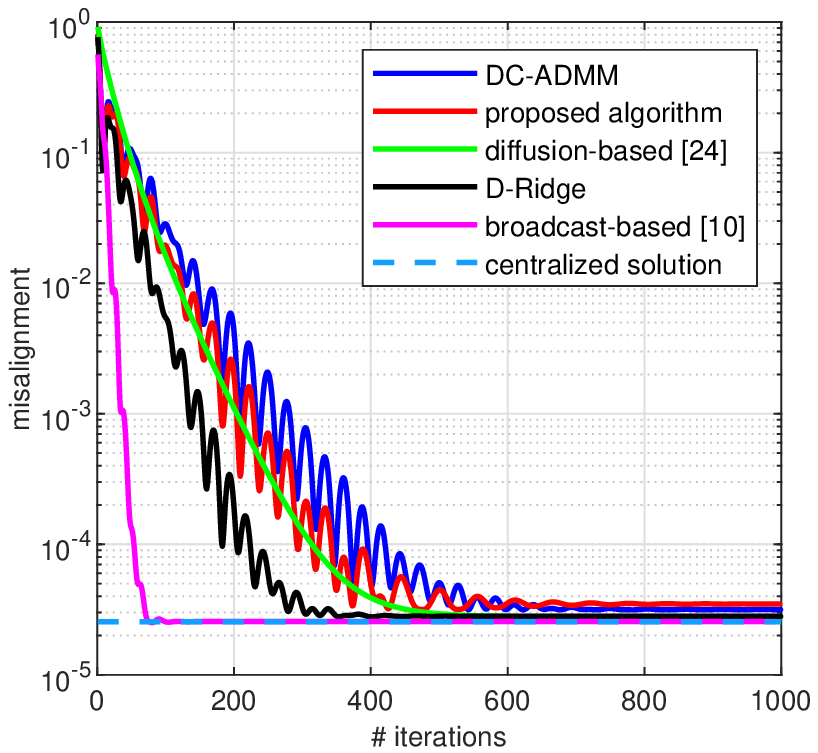}\label{fig:m200}}
\caption{The misalignment of the proposed algorithm and other considered algorithms for the ridge regression problems in different scenarios.}\label{fig:r1}
\end{figure*}

\begin{figure*}[htp]
  \centering
  \subfigure[Ridge regression with $N=20$, $M=200$, and $P_i=2$]{\includegraphics[scale=1.00]{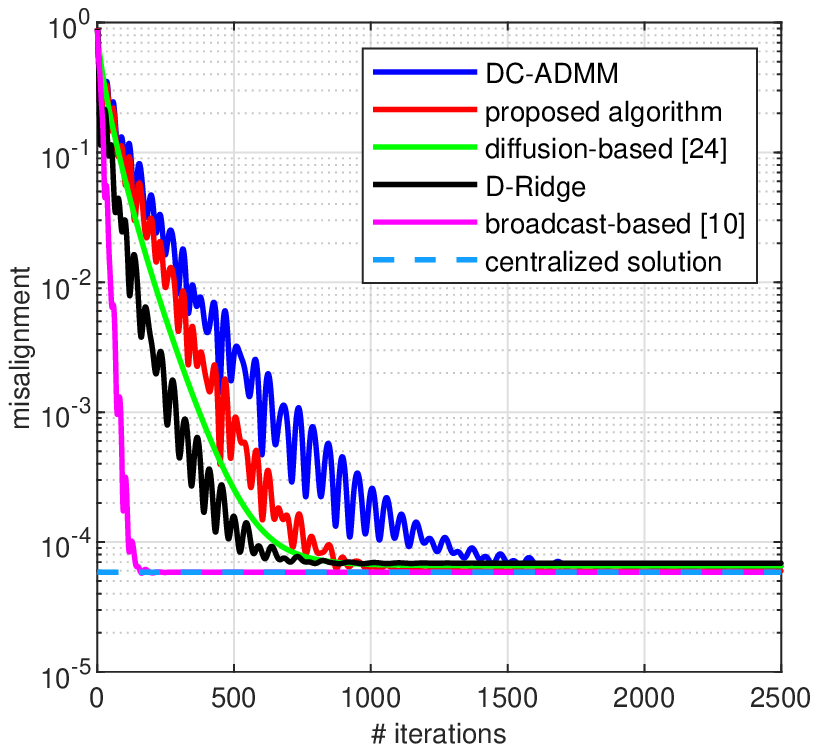}\label{fig:n20}}\quad
  \subfigure[Ridge regression with $N=10$, $M=200$, and $P_i=10$]{\includegraphics[scale=1.00]{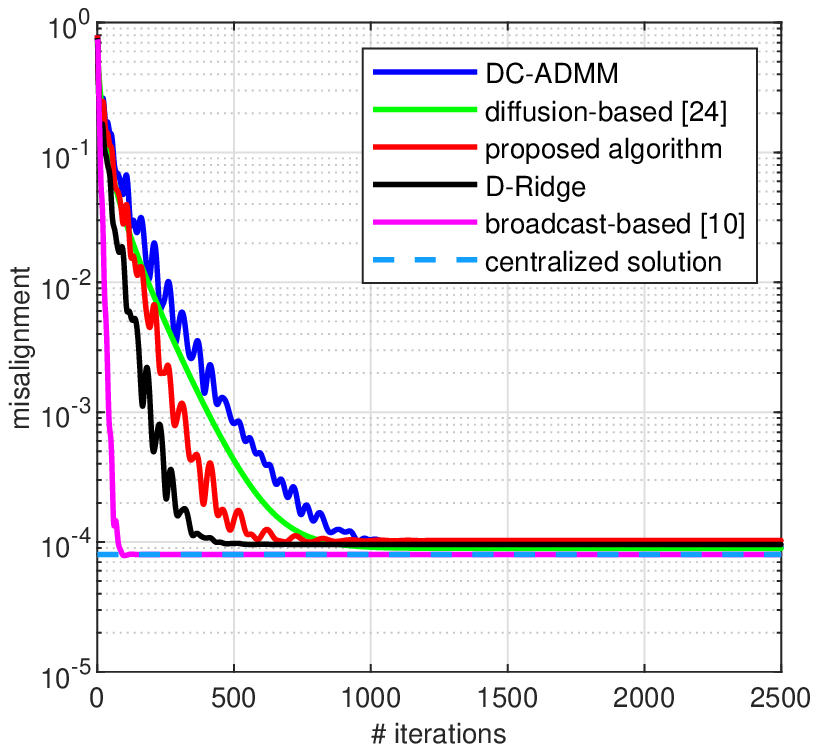}\label{fig:n10}}
\caption{The misalignment of the proposed algorithm and other considered algorithms for the ridge regression problems in different scenarios.}\label{fig:r2}
\end{figure*}

\begin{figure*}[htp]
  \centering
  \subfigure[Lasso regression with $N=10$, $M=50$, and $P_i=2$]{\includegraphics[scale=1.00]{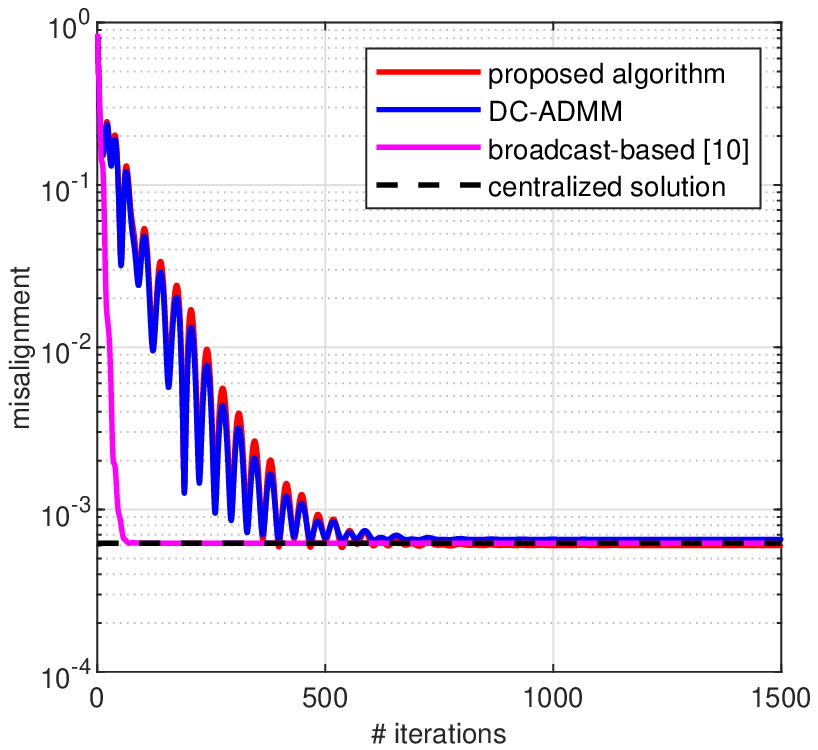}\label{fig:lm50}}\quad
  \subfigure[Lasso regression with $N=10$, $M=200$, and $P_i=2$]{\includegraphics[scale=1.00]{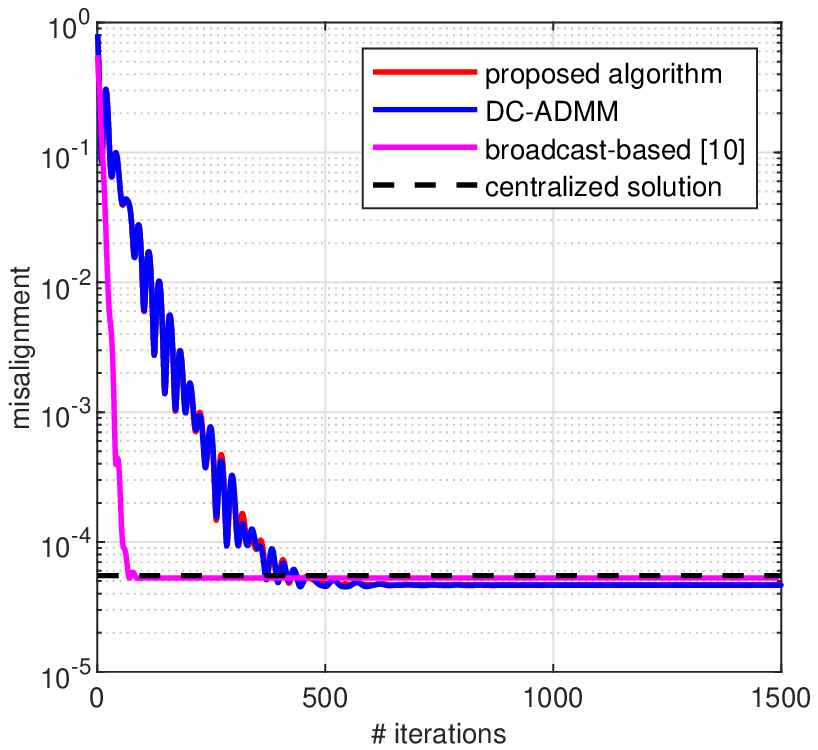}\label{fig:lm200}}
\caption{The misalignment of the proposed algorithm and other considered algorithms for the lasso regression problems in different scenarios.}\label{fig:l1}
\end{figure*}

\begin{figure*}[htp]
  \centering
  \subfigure[Lasso regression with $N=20$, $M=200$, and $P_i=2$]{\includegraphics[scale=1.00]{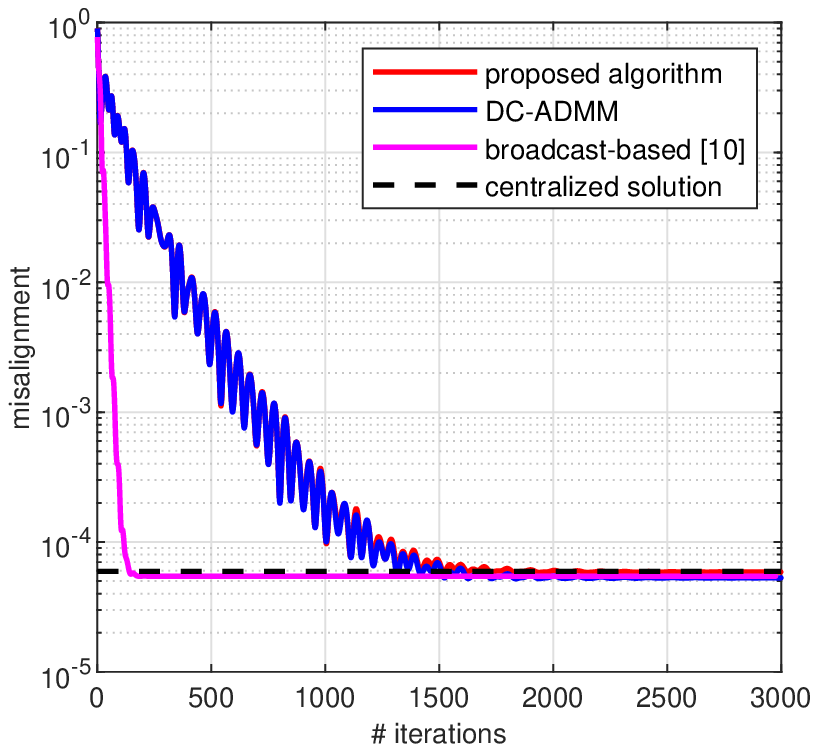}\label{fig:ln20}}\quad
  \subfigure[Lasso regression with $N=10$, $M=200$, and $P_i=10$]{\includegraphics[scale=1.00]{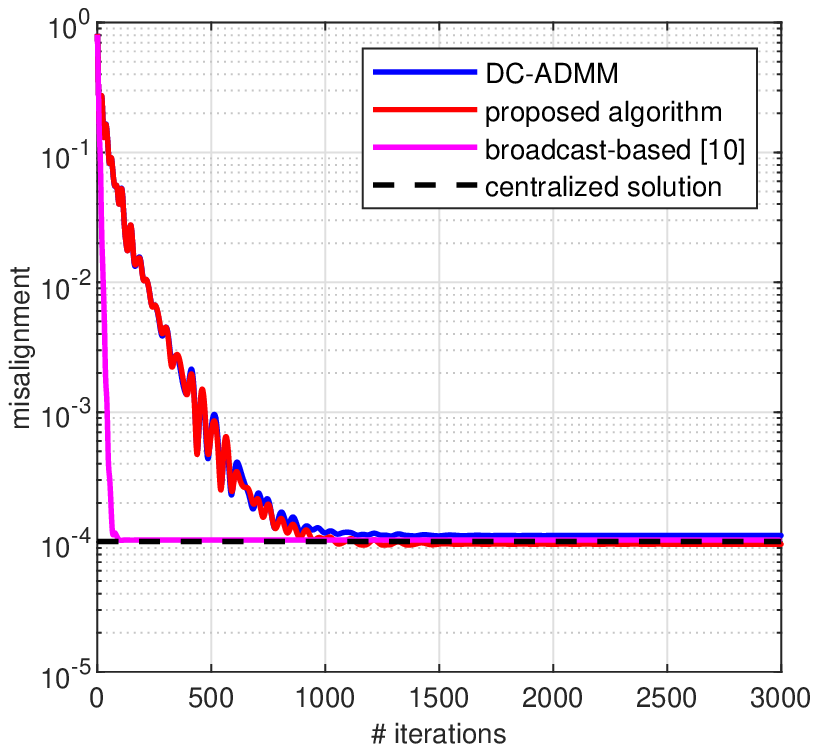}\label{fig:ln10}}
\caption{The misalignment of the proposed algorithm and other considered algorithms for the lasso regression problems in different scenarios.}\label{fig:l2}
\end{figure*}

In this section, we present some simulation results to evaluate the performance of the proposed algorithm. We first assess the proposed algorithm considering a distributed elastic-net regression problem with different numbers of local features, samples, and agents as well as different network topologies. Subsequently, we benchmark the proposed algorithm against the most relevant existing algorithms considering distributed ridge and lasso regression problems.

\subsection{Distributed Elastic-Net Regression}\label{elastic-net}

To evaluate the performance of the proposed algorithm in different scenarios, we consider the elastic-net regression problem. The calculation of the conjugate function for the objective function corresponding to this problem is practically infeasible. In the distributed setting, we solve the elastic-net regression problem by considering \begin{equation}
\begin{aligned}
f(\mathbf{x}_i) &= \norm{\sum_{i=1}^{N}\mathbf{A}_i\mathbf{x}_i-\mathbf{b}}^2\\
r_i(\mathbf{x}_i) &= \eta_1 \norm{\mathbf{x}_i}_1 + \eta_2 \norm{\mathbf{x}_i}^2
\label{eqn:elasticnet_sum_sim}
\end{aligned}
\end{equation}
where $\eta_1\in\mathbb{R}^+$ and $\eta_2\in\mathbb{R}^+$ are the regularization parameters. We calculate the response vector $\mathbf{b}$ as 
\begin{equation}
\mathbf{b}=\mathbf{A}\boldsymbol{\omega}+\boldsymbol{\psi}
\label{eq_resp_sim}
\end{equation}
where $\boldsymbol{\omega}\in\mathbb{R}^{P}$ and $\boldsymbol{\psi}\in\mathbb{R}^{M}$ are independently drawn from the multivariate normal distributions $\mathcal{N}(\mathbf{0},\mathbf{I}_{P})$ and $\mathcal{N}(\mathbf{0},0.1\mathbf{I}_M)$, respectively. We set the regularization parameters to $\eta_1=1$, $\eta_2=1$ and the penalty parameter to $\rho=2$. We use two iterations in the inner-loop BCD algorithm. We obtain the results by averaging over $100$ independent trials while considering a multi-agent network with a random topology where each agent links to three other agents on average. We evaluate the performance of the proposed algorithm using the misalignment metric that is defined as
$$\frac{\norm{\mathbf{x}^d(k)-\boldsymbol{\omega}}^2}{\norm{\boldsymbol{\omega}}^2}$$
where
$$\mathbf{x}^d(k) {=}\left[\mathbf{x}_1^{(k)\mathsf{T}},\hdots,\mathbf{x}_N^{(k)\mathsf{T}} \right]^{\mathsf{T}}$$
and $\mathbf{x}_i^{(k)}$ $\forall i\in \mathcal{V}$ denotes the local estimate at agent $i$.

In Fig.~\ref{fig:pi}, we plot the misalignment of the proposed algorithm versus its outer-loop iteration index for different values of $P_i$, i.e., $P_i=2$, $P_i=10$, $P_i=20$, and $P_i=50$ while $M=800$, $M=1000$, $M=1100$, and $M=1500$, respectively.
Fig.~\ref{fig:pi} shows that the proposed algorithm converges faster as the number of local features $P_i$ decreases. In Fig.~\ref{fig:m}, we set $P_i=2$ and use the same topology as in Fig.\ref{fig:pi} but consider different values of $M$. Fig.~\ref{fig:m} shows that the proposed algorithm achieves higher accuracy as the number of samples $M$ increases. Note that we include the misalignment of the centralized optimal solution in all figures.

In Fig.~\ref{fig:n}, we consider different values of $N$ while $P_i=2$, $M=500$, and the network topology is arbitrary but with an average node degree of three. Fig.~\ref{fig:n} shows that the proposed algorithm converges faster as the number of agents $N$ decreases. In Fig.~\ref{fig:net}, we evaluate the proposed algorithm by setting $N=10$, $P_i=2$, $M=500$ and considering four different common simple topologies, i.e.,
\begin{itemize}
    \item line: the agents are connected one after the other, hence, $|\mathcal{N}_i|=2$ for $1<i<N$ and $|\mathcal{N}_i|=1$ for $i=1$ and $i=N$
    \item ring: $|\mathcal{N}_i|=2$ for each $i\in\mathcal{V}$
    \item star: $|\mathcal{N}_i|=N-1$ for $i=1$ and $|\mathcal{N}_i|=1$ for $i=2,\hdots,N$
    \item fully-connected: each agent in the network is connected to all the other agents.
\end{itemize}
In Fig.~\ref{fig:net}, we observe that the proposed algorithm converges faster as the average number of links per agent increases, i.e., the average connectivity of the network increases.

\subsection{Distributed Ridge Regression}\label{ridge}

Considering a distributed ridge regression problem, in Figs.~\ref{fig:r1} and~\ref{fig:r2}, we benchmark the proposed algorithm against some existing baseline algorithms, namely, the broadcast-based algorithm for learning with distributed features proposed in \cite{Boyd2010}, the dual consensus ADMM (DC-ADMM) algorithm of \cite{dc_admm}, the consensus-based algorithm for ridge regression (D-Ridge) introduced in \cite{Grattonasilomar2018}, and the diffusion-based algorithm of \cite{Arablouei2015main}. The algorithms proposed in \cite{Grattonasilomar2018,Arablouei2015main} are only for solving the ridge regression problem. Here, we solve the problem \eqref{eqn:canon_objective} with the objective function \eqref{eqn:elasticnet_sum_sim} and set the $i$th agent's regularizer to
$$r_i(\mathbf{x}_i)=\eta \norm{\mathbf{x}_i}^2$$
where $\eta\in\mathbb{R}^+$ is the regularization parameter.

We calculate the response vector $\mathbf{b}$ as in \eqref{eq_resp_sim}. As per \cite{Grattonasilomar2018,Arablouei2015main}, we set the regularization parameter to $\eta=0.001$. We also set the number of inner-loop BCD iterations of the proposed algorithm to $2$ and obtain the results by averaging over $100$ independent trials. In Fig.~\ref{fig:m50}, we set $N=10$, $M=50$, and $P_i=2$. In Fig.~\ref{fig:m200}, the parameter setting is the same as Fig.~\ref{fig:m50} except for the number of samples $M$ being larger, i.e., $M=200$. In Fig.~\ref{fig:n20}, we keep $M=200$ and set the number of agents to $N=20$. In Fig.~\ref{fig:n10}, we set $N$ and $M$ to $10$ and $200$, respectively, while $P_i=10$ $\forall i\in\mathcal{V}$.

We observe in Figs.~\ref{fig:r1} and~\ref{fig:r2} that the proposed algorithm outperforms the DC-ADMM algorithm. It also perform competitively in comparison with the algorithms of \cite{Grattonasilomar2018,Arablouei2015main}, which are specifically tailored to the ridge regression problem. The superior performance of the broadcast-based algorithm of~\cite{Boyd2010} is due to its centralized processing. We include it here only as a reference.

\subsection{Distributed Lasso Regression}\label{lasso}

In Figs.~\ref{fig:l1} and~\ref{fig:l2}, we compare the performance of the proposed algorithm with that of the broadcast-based algorithm for learning with distributed features proposed in \cite{Boyd2010} and the DC-ADMM algorithm of \cite{dc_admm} considering a distributed lasso problem. Hence, we solve the problem \eqref{eqn:canon_objective} with the objective function \eqref{eqn:elasticnet_sum_sim} and set the $i$th agent's regularizer to
$$r_i(\mathbf{x}_i)=\eta \norm{\mathbf{x}_i}_1$$
where $\eta\in\mathbb{R}^+$ is the regularization parameter.

We calculate the response vector $\mathbf{b}$ as in \eqref{eq_resp_sim}. As per \cite{Grattonasilomar2018,Arablouei2015main}, we set the regularization parameter to $\eta=0.001$. We also set the number of inner-loop BCD iterations of the proposed algorithm to $2$ and obtain the results by averaging over $100$ independent trials. In Fig.~\ref{fig:lm50}, we set $N=10$, $M=50$, and $P_i=2$. In Fig.~\ref{fig:lm200}, the parameter setting is the same as Fig.~\ref{fig:lm50} except for the number of samples $M$ being larger, i.e., $M=200$. In Fig.~\ref{fig:ln20}, we keep $M=200$ and set the number of agents to $N=20$. In Fig.~\ref{fig:ln10}, we set $N$ and $M$ to $10$ and $200$, respectively, while $P_i=10$ $\forall i\in\mathcal{V}$.

We observe in Figs.~\ref{fig:l1} and~\ref{fig:l2} that the proposed algorithm performs very similar to the DC-ADMM algorithm as the learning curves of the two algorithms almost overlap. Again, the superior performance of the broadcast-based algorithm of~\cite{Boyd2010} is due to its centralized processing.

\subsection{Discussion}

The main advantage of the proposed algorithm is in its ability to solve generic feature-partitioned distributed optimization problems without resorting to any conjugate function even when the objective function is non-smooth. This is unique to our proposed algorithm and, to the best of our knowledge, there is no existing algorithm with the same utility. That is why we do not compare the proposed algorithm with any other existing algorithm in Section~\ref{elastic-net} where the problem at hand is feature-distributed elastic-net regression. The existing algorithms for feature-partitioned distributed optimization such as DC-ADMM require the conjugate function of the objective or regularization function. In the case of elastic-net regression, calculating the conjugate function is impracticable.

The simulation results in Sections~\ref{ridge} and~\ref{lasso} are to provide a comparative study of the performance of the proposed algorithm with respect to the other most relevant existing algorithms. As evident by the results, the proposed algorithm's performance in solving the distributed ridge and lasso regression problems is on par with those of its state-of-the-art competitors, even those that have specifically been design to solve these problems.

As seen in the figures, in all simulations, the network-wide average estimate of the proposed algorithm converges to the corresponding optimal centralized solution. Although not shown here for conciseness, we have observed that the estimates at all agents also converge to the optimal solution in all the experiments corroborating our theoretical findings in Section~\ref{conv-anal}.

In all simulations, we utilize only two BCD iterations with no extra inter-agent communication overhead. Therefore, the computational complexity and communication requirements of the proposed algorithm are of the same order as those of the related existing algorithms such as DC-ADMM. Indeed, we did not observe any significant difference in the per-iteration run time of the proposed and DC-ADMM algorithms.



\section{Conclusion}

We proposed a distributed algorithm for learning with non-smooth objective functions under distributed features. We reformulated the considered non-separable problem into a dual form that is separable and solved it via the ADMM. Subsequently, we devised an approach based on articulating the dual of the dual problem to overcome the challenge of computing the involved conjugate functions, which may be hard or even infeasible with some objective functions. We employed the BCD algorithm to solve the dual of the dual problem. Therefore, unlike most existing algorithms for solving learning problems with feature partitioning, the proposed algorithm does not require the explicit calculation of any conjugate of the objective function. We verified the convergence of the proposed algorithm to the optimal solution through both theoretical analysis and numerical simulations.

\bibliographystyle{IEEEtran}
\bibliography{IEEEabrv,references}

\begin{thebibliography}{10}
\providecommand{\url}[1]{#1}
\csname url@samestyle\endcsname
\providecommand{\newblock}{\relax}
\providecommand{\bibinfo}[2]{#2}
\providecommand{\BIBentrySTDinterwordspacing}{\spaceskip=0pt\relax}
\providecommand{\BIBentryALTinterwordstretchfactor}{4}
\providecommand{\BIBentryALTinterwordspacing}{\spaceskip=\fontdimen2\font plus
\BIBentryALTinterwordstretchfactor\fontdimen3\font minus
  \fontdimen4\font\relax}
\providecommand{\BIBforeignlanguage}[2]{{%
\expandafter\ifx\csname l@#1\endcsname\relax
\typeout{** WARNING: IEEEtran.bst: No hyphenation pattern has been}%
\typeout{** loaded for the language `#1'. Using the pattern for}%
\typeout{** the default language instead.}%
\else
\language=\csname l@#1\endcsname
\fi
#2}}
\providecommand{\BIBdecl}{\relax}
\BIBdecl

\bibitem{Gratton2021eusipco}
C.~Gratton, N.~K.~D. Venkategowda, R.~Arablouei, and S.~Werner, ``Distributed
  learning over networks with non-smooth regularizers and feature
  partitioning,'' in \emph{Proc. European Speech and Signal Processing
  Conference}, Aug. 2021.

\bibitem{Mingyihongbook}
Z.~Han, M.~Hong, and D.~Wang, \emph{Signal processing and networking for big
  data applications}.\hskip 1em plus 0.5em minus 0.4em\relax Cambridge
  University Press, 2017.

\bibitem{Grattonasilomar2018}
C.~Gratton, N.~K.~D. Venkategowda, R.~Arablouei, and S.~Werner, ``Distributed
  ridge regression with feature partitioning,'' in \emph{Proc. Asilomar
  Conference on Signals, Systems, and Computers}, Oct. 2018.

\bibitem{Gratton2019}
------, ``Consensus-based distributed total least-squares estimation using
  parametric semidefinite programming,'' in \emph{Proc. IEEE International
  Conference on Acoustics, Speech and Signal Processing}, May 2019, pp.
  5227--5231.

\bibitem{Giannakis2016}
G.~B. Giannakis, Q.~Ling, G.~Mateos, and I.~D. Schizas, \emph{Splitting Methods
  in Communication, Imaging, Science, and Engineering}, ser. Scientific
  Computation, R.~Glowinski, S.~J. Osher, and W.~Yin, Eds.\hskip 1em plus 0.5em
  minus 0.4em\relax Cham: Springer International Publishing, 2016.

\bibitem{Hajinezhad2019}
D.~Hajinezhad, M.~Hong, and A.~Garcia, ``{ZONE}: Zeroth-order nonconvex
  multiagent optimization over networks,'' \emph{IEEE Transactions on Automatic
  Control}, vol.~64, no.~10, pp. 3995--4010, Oct. 2019.

\bibitem{Nedic2009}
A.~Nedic and A.~Ozdaglar, ``Distributed subgradient methods for multi-agent
  optimization,'' \emph{IEEE Transactions on Automatic Control}, vol.~54,
  no.~1, pp. 48--61, Jan. 2009.

\bibitem{Gratton2020eusipco}
C.~{Gratton}, N.~K.~D. {Venkategowda}, R.~{Arablouei}, and S.~{Werner},
  ``Distributed learning with non-smooth objective functions,'' in \emph{Proc.
  28th European Signal Processing Conference}, Jan. 2021, pp. 2180--2184.

\bibitem{Bertrand2011}
A.~Bertrand and M.~Moonen, ``Consensus-based distributed total least squares
  estimation in ad hoc wireless sensor networks,'' \emph{IEEE Transactions on
  Signal Processing}, vol.~59, no.~5, pp. 2320--2330, May 2011.

\bibitem{Boyd2010}
S.~Boyd, N.~Parikh, E.~Chu, B.~Peleato, and J.~Eckstein, ``Distributed
  optimization and statistical learning via the alternating direction method of
  multipliers,'' \emph{Foundations and Trends in Machine Learning}, vol.~3,
  no.~1, pp. 1--122, Jan. 2010.

\bibitem{Ying2019}
B.~Ying, K.~Yuan, and A.~H. Sayed, ``Supervised learning under distributed
  features,'' \emph{IEEE Transactions on Signal Processing}, vol.~67, no.~4,
  pp. 977--992, Feb. 2019.

\bibitem{Zheng2011}
H.~Zheng, S.~R. Kulkarni, and H.~V. Poor, ``Attribute-distributed learning:
  Models, limits, and algorithms,'' \emph{IEEE Transactions on Signal
  Processing}, vol.~59, no.~1, pp. 386--398, Jan. 2011.

\bibitem{Mangasarian}
O.~L. Mangasarian, E.~W. Wild, and G.~M. Fung, ``Privacy-preserving
  classification of vertically partitioned data via random kernels,'' \emph{ACM
  Transactions on Knowledge Discovery from Data}, vol.~2, no.~3, Oct. 2008.

\bibitem{Vaidya}
J.~Vaidya and C.~Clifton, ``Privacy-preserving k-means clustering over
  vertically partitioned data,'' in \emph{Proc. 9th ACM International
  Conference on Knowledge Discovery and Data Mining}, 2003, pp. 206--215.

\bibitem{Mota2012}
J.~F.~C. Mota, J.~M.~F. Xavier, P.~M.~Q. Aguiar, and M.~Puschel, ``Distributed
  basis pursuit,'' \emph{IEEE Transactions on Signal Processing}, vol.~60,
  no.~4, pp. 1942--1956, Apr. 2012.

\bibitem{Leus2018}
C.~{Manss}, D.~{Shutin}, and G.~{Leus}, ``Distributed splitting-over-features
  sparse bayesian learning with alternating direction method of multipliers,''
  in \emph{2018 IEEE International Conference on Acoustics, Speech and Signal
  Processing}, 2018, pp. 3654--3658.

\bibitem{Mota2013}
J.~F.~C. Mota, J.~M.~F. Xavier, P.~M.~Q. Aguiar, and M.~Puschel, ``D-admm: A
  communication-efficient distributed algorithm for separable optimization,''
  \emph{IEEE Transactions on Signal Processing}, vol.~61, no.~10, pp.
  2718--2723, May 2013.

\bibitem{Kashyap2016}
N.~Kashyap, S.~Werner, Y.-F. Huang, and R.~Arablouei, ``Privacy preserving
  decentralized power system state estimation with phasor measurement units,''
  in \emph{Proc. 2016 IEEE Sensor Array and Multichannel Signal Processing
  Workshop}, Jul. 2016, pp. 1--5.

\bibitem{Heinze2015}
\BIBentryALTinterwordspacing
C.~Heinze{-}Deml, B.~McWilliams, N.~Meinshausen, and G.~Krummenacher, ``{LOCO}:
  Distributing ridge regression with random projections,'' 2015. [Online].
  Available: \url{http://arxiv.org/pdf/1406.3469}
\BIBentrySTDinterwordspacing

\bibitem{dualloco}
C.~Heinze, B.~McWilliams, and N.~Meinshausen, ``{DUAL-LOCO}: Distributing
  statistical estimation using random projections,'' in \emph{Proc. 19th
  International Conference on Artificial Intelligence and Statistics}, vol.~51,
  May 2016, pp. 875--883.

\bibitem{Heinze2017}
\BIBentryALTinterwordspacing
C.~Heinze{-}Deml, B.~McWilliams, and N.~Meinshausen, ``Preserving differential
  privacy between features in distributed estimation,'' 2017. [Online].
  Available: \url{http://arxiv.org/abs/1703.00403}
\BIBentrySTDinterwordspacing

\bibitem{Sayed_dictionary}
J.~{Chen}, Z.~J. {Towfic}, and A.~H. {Sayed}, ``Dictionary learning over
  distributed models,'' \emph{IEEE Transactions on Signal Processing}, vol.~63,
  no.~4, pp. 1001--1016, 2015.

\bibitem{Sayed_eusipco}
S.~A. {Alghunaim}, M.~{Yan}, and A.~H. {Sayed}, ``A multi-agent primal-dual
  strategy for composite optimization over distributed features,'' in
  \emph{Proc. 28th European Signal Processing Conference}, Jan. 2021, pp.
  2095--2099.

\bibitem{Arablouei2015main}
R.~Arablouei, K.~Do{\u{g}}an{\c{c}}ay, S.~Werner, and Y.-F. Huang,
  ``Model-distributed solution of regularized least-squares problem over sensor
  networks,'' in \emph{Proc. 2015 IEEE International Conference on Acoustics,
  Speech and Signal Processing}, Apr. 2015, pp. 3821--3825.

\bibitem{Virginiasmith}
V.~Smith, S.~Forte, C.~Ma, M.~Tak\'{a}\v{c}, M.~I. Jordan, and M.~Jaggi,
  ``{C}o{C}o{A}: A general framework for communication-efficient distributed
  optimization,'' \emph{J. Mach. Learn. Res.}, vol.~18, no.~1, p. 8590–8638,
  Jan. 2017.

\bibitem{dc_admm}
T.~{Chang}, M.~{Hong}, and X.~{Wang}, ``Multi-agent distributed optimization
  via inexact consensus {ADMM},'' \emph{IEEE Transactions on Signal
  Processing}, vol.~63, no.~2, pp. 482--497, 2015.

\bibitem{CISS2018}
B.~{Zhang}, J.~{Geng}, W.~{Xu}, and L.~{Lai}, ``Communication efficient
  distributed learning with feature partitioned data,'' in \emph{2018 52nd
  Annual Conference on Information Sciences and Systems (CISS)}, Mar. 2018, pp.
  1--6.

\bibitem{Diniu2019}
Y.~Hu, D.~Niu, J.~Yang, and S.~Zhou, ``{FDML}: A collaborative machine learning
  framework for distributed features,'' in \emph{Proceedings of the 25th ACM
  SIGKDD International Conference on Knowledge Discovery \& Data Mining}, 2019,
  pp. 2232--2240.

\bibitem{Szurley2017}
J.~Szurley, A.~Bertrand, and M.~Moonen, ``Topology-independent distributed
  adaptive node-specific signal estimation in wireless sensor networks,''
  \emph{IEEE Transactions on Signal and Information Processing over Networks},
  vol.~3, no.~1, pp. 130--144, 2017.

\bibitem{Chen2014}
J.~Chen, C.~Richard, and A.~H. Sayed, ``Diffusion {LMS} for clustered multitask
  networks,'' in \emph{2014 IEEE International Conference on Acoustics, Speech
  and Signal Processing}, 2014, pp. 5487--5491.

\bibitem{Berberidis2014}
N.~Bogdanović, J.~Plata-Chaves, and K.~Berberidis, ``Distributed
  incremental-based {LMS} for node-specific adaptive parameter estimation,''
  \emph{IEEE Transactions on Signal Processing}, vol.~62, no.~20, pp.
  5382--5397, 2014.

\bibitem{Berberidis2015}
J.~Plata-Chaves, N.~Bogdanović, and K.~Berberidis, ``Distributed
  diffusion-based {LMS} for node-specific adaptive parameter estimation,''
  \emph{IEEE Transactions on Signal Processing}, vol.~63, no.~13, pp.
  3448--3460, 2015.

\bibitem{Boyd_clustering_graphs}
D.~Hallac, J.~Leskovec, and S.~Boyd, ``Network lasso: Clustering and
  optimization in large graphs,'' in \emph{Proceedings of the 21th ACM SIGKDD
  International Conference on Knowledge Discovery and Data Mining}, 2015, p.
  387–396.

\bibitem{Scaglione}
T.~{Chang}, A.~{Nedić}, and A.~{Scaglione}, ``Distributed constrained
  optimization by consensus-based primal-dual perturbation method,'' \emph{IEEE
  Transactions on Automatic Control}, vol.~59, no.~6, pp. 1524--1538, 2014.

\bibitem{Palomar_utility}
D.~P. {Palomar} and {Mung Chiang}, ``A tutorial on decomposition methods for
  network utility maximization,'' \emph{IEEE Journal on Selected Areas in
  Communications}, vol.~24, no.~8, pp. 1439--1451, 2006.

\bibitem{Fukushima1992}
M.~Fukushima, ``Application of the alternating direction method of multipliers
  to separable convex programming problems,'' \emph{Computational Optimization
  and Applications}, vol.~1, pp. 93--111, 1992.

\bibitem{Bertsekas99}
D.~Bertsekas, \emph{Nonlinear programming}.\hskip 1em plus 0.5em minus
  0.4em\relax Athena Scientific, 1999.

\bibitem{Boyd2014}
N.~Parikh and S.~Boyd, ``Proximal algorithms,'' \emph{Foundations and Trends in
  Optimization}, vol.~1, no.~3, p. 127–239, Jan. 2014.

\bibitem{BoydStephenP2004Co}
S.~Boyd and L.~Vandenberghe, \emph{Convex optimization}.\hskip 1em plus 0.5em
  minus 0.4em\relax Cambridge University Press, 2004.

\bibitem{Borwein}
J.~M. Borwein and A.~S. Lewis, \emph{Convex Analysis and Nonlinear
  Optimization, Theory and Examples}.\hskip 1em plus 0.5em minus 0.4em\relax
  Springer, 2000.

\end{thebibliography}
\end{document}